\documentclass[11pt,letterpaper]{article}
\usepackage{style}
\begin{document}

\title{Constant-round quantum advantage in communication complexity for total functions}
\author{
Atsuya Hasegawa\\
Graduate School of Mathematics\\
Nagoya University \\
\texttt{atsuya.hasegawa@math.nagoya-u.ac.jp}
\and Fran{\c c}ois Le Gall\\
Graduate School of Mathematics\\
Nagoya University\\
\texttt{legall@math.nagoya-u.ac.jp}
}
\date{}
\maketitle
\begin{abstract}
We show that there exists a total function for which there is a polynomial gap between the randomized and the constant-round quantum communication complexity. Previously, such a separation was known only for quantum protocols using polynomially many rounds.

\end{abstract}

\section{Introduction}

\subsection{Background}

Communication complexity~\cite{YaoSTOC79,YaoFOCS93} provides a basic setting for studying the power of quantum communication.  Two parties, Alice and Bob, receive inputs $x,y \in \{0,1\}^n$ and want to compute a function $f(x,y)$ using as little communication as possible. A central question is to quantify the advantage of quantum communication over classical communication. The case of total functions is particularly fundamental, as it provides a more natural setting in which the protocol must work on every pair of inputs, without any promise on the inputs.

The first asymptotic quantum advantage for a total function was shown in the seminal work by Buhrman, Cleve, and Wigderson~\cite{Buhrman+STOC1998} for the disjointness function.  For $\DISJ_n(x,y)=1\Leftrightarrow x_i y_i=0 \text{ for every } i\in[n]$, their protocol runs Grover search~\cite{GroverSTOC1996} on the virtual input $z=x\wedge y$, whose $i$-th bit $z_i=x_i\wedge y_i$ is jointly determined by Alice and Bob.  To implement a quantum query $|i,b\rangle\mapsto |i,b\oplus z_i\rangle$ on a superposition of indices, Alice coherently computes $x_i$ into an ancilla and sends the index, answer, and ancilla registers to Bob.  Bob uses $y_i$ to flip the answer qubit precisely when $x_i=y_i=1$, sends the registers back to Alice, and Alice uncomputes the ancilla.  Thus, each Grover query is implemented through an Alice-to-Bob-to-Alice exchange of $O(\log n)$ qubits.  Since Grover search makes $O(\sqrt n)$ queries, this gives a quantum protocol using $O(\sqrt n\log n)$ communication.

Classical bounded-error protocols require linear communication, $\RC{\DISJ_n}=\Theta(n)$, as shown in~\cite{KalyanasundaramSchintger1992,Razborov1992}. The quantum upper bound was subsequently improved by H{\o}yer and de Wolf~\cite{Hoyer+STACS02}, and Aaronson and Ambainis~\cite{AaronsonAmbainisToC2005} obtained the optimal bound $\QC{\DISJ_n}=O(\sqrt n)$, matching Razborov's lower bound~\cite{Razborov2003}.  Thus, disjointness exhibits a quadratic separation between randomized and quantum communication complexity.

The cheat sheet framework of Aaronson, Ben-David, and Kothari~\cite{Aaronson+STOC16} made it possible to go beyond the quadratic Grover-type separation in query complexity.  Anshu et al.~\cite{Anshu+FOCS2016} developed a communication analogue, and obtained the first super-quadratic separation for a total communication problem.  Subsequent work of Bansal and Sinha~\cite{BansalSinhaSTOC2021} and Sherstov, Storozhenko, and Wu~\cite{Sherstov+SICOMP23} strengthened the underlying randomized-versus-quantum query separations.  These improved query separations can be transferred to communication complexity by combining randomized lifting theorems such as~\cite{Chattopadhyay+ICAP19} for the randomized lower bounds with the query-to-communication simulation of Buhrman, Cleve, and Wigderson~\cite{Buhrman+STOC1998} for the quantum upper bounds.

The resulting quantum protocols are, however, highly interactive. In the standard measure of two-way quantum communication complexity, only the total amount of communication is counted, while the number of rounds is unrestricted.  As illustrated above for disjointness, the BCW simulation implements each quantum query through an Alice-to-Bob-to-Alice exchange and hence uses two rounds per query. The later protocols for disjointness and functions based on the cheat sheet technique follow the same distributed-query paradigm. Consequently, their round complexity scales with the sequential query complexity of the underlying quantum algorithms, which is polynomial in the input length in all previously known constructions.

However, implementing polynomially many rounds of quantum communication seems difficult with quantum technology in the near future. This motivates the following question:

\begin{center}
\emph{Can we obtain asymptotic quantum advantage in communication complexity for total functions \\ using only a constant number of rounds of quantum communication?}
\end{center}

\subsection{Our contribution}

We answer the question above affirmatively. To the best of our knowledge, this gives the first polynomial quantum advantage for a total function using only a constant number of quantum messages. We denote by $\QCr{f}{r}$ the $r$-round quantum communication complexity for $f$.

\begin{theorem}[Informal version of \cref{th:main-lb,th:main-ub}]\label{th:main-informal}
For every fixed integer $t\ge 1$, there exists a family of total
Boolean functions $f:\{0,1\}^N\times\{0,1\}^N\to\{0,1\}$
such that
\[
    \RC{f}
    =
    \widetilde{\Omega}\left(
        \left(\QCr{f}{2t+2}\right)^{
            \frac32-\frac{1}{4t}
        }
    \right).
\]
\end{theorem}

Our result is a power $3/2-1/(4t)$ separation.  In particular, setting $t=1$ gives a 4-round protocol and a power $5/4$ separation. 
More generally, for every constant $\varepsilon>0$, choosing a sufficiently large but fixed $t$ gives a constant-round separation of power $3/2-\varepsilon$. Moreover, our $2t+2$-round communication protocol consists of $2t$-round quantum communication and subsequent $2$-round classical communication. In particular, our 4-round communication involves only $2$-round quantum communication.

As an immediate consequence of \cref{th:main-informal}, we also obtain an exponential round separation in bandwidth-limited distributed computing.  On a two-node network, a $T$-round algorithm for the CONGEST model with bandwidth $B$ can be viewed as a two-party communication protocol with $T$ rounds of $B$ (qu)bits communication.

\begin{corollary}[Informal version of \cref{cor:congest}]\label{cor:congest-informal} For every fixed integer $t\ge 1$, in the two-node $\mathrm{CONGEST}$ model with bandwidth $B=\widetilde{\Theta}(n^2)$, there exists a family of total Boolean functions that can be computed in at most $2t+2$ rounds in the quantum setting, whereas every bounded-error
randomized classical algorithm requires
\[
    \widetilde{\Omega}
    \left(
        n^{1-\frac{1}{2t}}
    \right)
\]
rounds.
\end{corollary}

In particular, when $t=1$, this gives a 4-round quantum algorithm and an $\widetilde{\Omega}(\sqrt{n})$ classical round lower bound.

\subsection{Technical overview}

We use the partial function $\Ff_n$ from Sherstov, Storozhenko, and Wu~\cite{Sherstov+SICOMP23} (or equivalently the corresponding result of Bansal and Sinha~\cite{BansalSinhaSTOC2021}). For every fixed $t$, this function can be computed with $t$ quantum queries and constant advantage, while its randomized query complexity is
\[
    \widetilde{\Omega}\left( n^{1-\frac{1}{2t}} \right).
\]
Following the cheat sheet construction of Aaronson, Ben-David, and Kothari~\cite{Aaronson+STOC16}, we then compose $\Ff_n$ with the $\ANDOR$ tree and add a cheat sheet array, obtaining a total query function $\Gg_n$. Then, the randomized query lower bound of the function is
\[
    \RQ{\Gg_n} = \widetilde{\Omega}\left( n^{3-\frac{1}{2t}} \right).
\]
We finally compose every input bit of $\Gg_n$ with an inner-product gadget $\IP_m$, where $m=\Theta(\log n)$. The randomized lifting theorem of~\cite{Chattopadhyay+ICAP19} transfers the query lower bound to randomized communication complexity. The construction of $\Gg_n\circ\IP_m$ and its randomized lower bound follow the standard cheat sheet and lifting framework. Our contribution is the constant-message quantum protocol described as follows.

For our quantum upper bound, let us consider one of the $10\log n$ instances of $\Ff_n\circ\ANDOR\circ\IP_m$. The $j$-th input bit seen by $\Ff_n$ is the value of an $\ANDOR$ tree whose leaves are inner-product gadgets distributed between Alice and Bob. To simulate a coherent query to this bit, Alice sends the query index, the target qubit, and a coherent copy of her entire block for the queried $\ANDOR$ instance.  Bob computes all relevant inner products, evaluates the $\ANDOR$ tree into the target qubit, and returns the registers. Alice then uncomputes her block. This is a direct BCW-style simulation: by sending her entire block, Alice enables Bob to evaluate the inner $\ANDOR$ function exactly within a single Alice-to-Bob-to-Alice exchange. If $Q(\Ff_n)=t$, Alice can obtain the address of the cheat sheet using $2t$ rounds of $\widetilde{O}(n^2)$-qubit communication.

After the address has been computed, Alice and Bob read and check the content of the corresponding cell of the cheat sheet using 2 rounds of classical messages. Alice sends Bob her inner-product blocks for the addressed cheat sheet cell, allowing Bob to reconstruct its contents. Bob then sends Alice the certificates together with his shares of the input bits referred to by them. Alice verifies the $\ANDOR$ certificates against the reconstructed input bits, thereby checking both their claimed values and the promise condition for $\Ff_n$.

\subsection{Related work}

Our work is closely related to two lines of research. The first is the study of parallel (non-adaptive) quantum query complexity. Carolan, Gilani, and Vempati~\cite{Carolan+ITCS25} used the cheat sheet framework to obtain a quantum advantage with a constant number of rounds of parallel queries. Their motivation is different from ours: they study the power of parallel quantum queries, whereas we ask whether a quantum advantage in communication complexity requires many rounds of interaction. Nevertheless, the resulting algorithms share a similar high-level structure. In both settings, many quantum queries are arranged into a constant number of sequential stages, with all queries within each stage performed in parallel. In our communication setting, all queries in one such stage are simulated together through a single Alice-to-Bob-to-Alice exchange. Consequently, the number of communication rounds is determined by the number of parallel query stages rather than by the total number of queries. Our protocol then uses two additional classical messages to read the addressed cheat sheet cell and verify the certificates contained in it.

The second is the study of bounded-round quantum communication complexity for disjointness. Jain, Radhakrishnan, and Sen~\cite{Jain+FOCS2003} showed that any $r$-round quantum protocol for $\DISJ_n$ requires $\Omega(n/r^2)$ communication. Braverman et al.~\cite{BravermanSICOMP2018} subsequently improved this bound to $\widetilde{\Omega}(n/r+r)$, obtaining a near-optimal round and communication tradeoff.  In particular, any constant-round quantum protocol for $\DISJ_n$ requires $\Omega(n)$ communication. Thus, the unrestricted-round quadratic quantum advantage for disjointness disappears in the constant-round setting. Our result shows that this phenomenon is not universal: there exists a total function that exhibits a polynomial quantum advantage even when the number of quantum rounds is bounded by a constant.

\subsection{Concurrent work}

Independently and concurrently with our work, Gavinsky~\cite{Gavinsky2026} showed a total function that admits a 2-round quantum protocol with polylogarithmic communication complexity and requires polynomial randomized communication complexity, which is stronger than our result. 

Gavinsky started with a problem called ``Shape'' for which quantum one-way communication complexity is polynomial in $\log n$, and classical two-way communication complexity is $\Omega(\poly (n))$ \cite{Gavinsky2020}. Then, he made it a total communication function by applying the lookup method introduced by Anshu et al. \cite{Anshu+FOCS2016} as a communication analog of the cheat sheet technique. Since its quantum one-way communication complexity is $\poly (\log (n))$, the address of the cheat sheet can be obtained by one-way $\poly (\log (n))$ qubit communication from Alice to Bob. In the look-up method, the cheat sheet is distributed to Alice and Bob so that their bit-wise XOR is the cheat sheet. Applying a fully linear PCP construction by Boneh et al. \cite{boneh2019zero}, Gavinsky showed that the steps to read and check the certificates can be done by $\poly(\log (n))$-bit one-way communication from Bob to Alice! See also \cite{gall2026improved} for its refinement.
\section{Preliminaries}

We refer to \cite{KushilevitzNoam1996,Buhrman+Review2010,RaoYehudayoff2020} for references in classical and quantum communication complexity.

\paragraph{Notations.}

For a possibly partial Boolean function $f:D\to\{0,1\}$, where $D\subseteq\{0,1\}^N$, we write
\begin{itemize}
    \item $\RQ{f}$ for its bounded-error randomized query complexity;
    \item $\QQ{f}$ for its bounded-error quantum query complexity.
\end{itemize}
For a possibly partial communication problem $F:D\to\{0,1\}$, where $D\subseteq\mathcal X\times\mathcal Y$, we write
\begin{itemize}
    \item $\RC{F}$ for its bounded-error randomized communication
    complexity;
    \item $\QC{F}$ for its bounded-error quantum communication
    complexity.
\end{itemize}

In this paper, we consider the quantum $r$-round bounded-error communication complexity, which is defined as follows.

\begin{definition}[Quantum $r$-round bounded-error communication complexity]
Let $F:\mathcal X\times \mathcal Y \to \{0,1\}$ be a possibly partial Boolean function. An $r$-round quantum communication protocol for $F$ is a protocol between Alice and Bob in which Alice receives $x\in\mathcal X$, Bob receives $y\in\mathcal Y$, and they exchange at most $r$ quantum messages, starting with a message from Alice to Bob. Between messages, each player may apply quantum operations to their private registers, depending on their own input and the messages received so far. The players do not share any prior entanglement. The communication cost of the protocol is the total number of qubits exchanged.

At the end of the protocol, one of the players outputs a bit $\Pi(x,y)$.\footnote{We allow either party to produce the final output. Requiring both parties to know the output would cost at most one additional one-bit message.} We say that the protocol computes $F$ with error at most $1/3$ if, for every $(x,y)\in\operatorname{Dom}(F)$,
\[
\Pr[\Pi(x,y)=F(x,y)]\ge 2/3.
\]
The $r$-round bounded-error quantum communication complexity of $F$, denoted by $\QCr{F}{r}$, is the minimum communication cost among all $r$-round quantum communication protocols that compute $F$ with error at most $1/3$.
\end{definition}

\paragraph{The two-node $\mathrm{CONGEST}(B)$ model.}

The two-node $\mathrm{CONGEST}(B)$ model can be viewed as two-party communication in which each party may send at most $B$ bits to the other per round.  In the quantum version, bits are replaced by qubits, and no prior entanglement is allowed. Hence, a $T$-round algorithm has a communication cost of at most $2TB$.

\paragraph{Query complexity separation using cheat sheets.}
We first review the results from \cite{Aaronson+STOC16}.

Let
\[
    \Ff_n:D_n\to\{0,1\},
    \qquad
    D_n\subseteq\{0,1\}^n,
\]
be a partial function exhibiting a quantum advantage (e.g. the Forrelation function \cite{Aaronson+SICOMP18}). 
The first idea is to combine this function with the AND-OR tree function $\ANDOR \colon \{0,1\}^{n^2} \to\{0,1\}$. This gives the (partial) function $\Ff_n\circ \ANDOR\colon \{0,1\}^{n^3} \to\{0,1\}$. Using the fact that the value of $\ANDOR$ can be certified by reading only $n$ bits of the input, we construct a cheat sheet to make the function total. 

Concretely, consider $10\log n$ independent instances of $\Ff_n\circ \ANDOR$. Introduce as an additional input an array called the ``cheat sheet'' consisting of $n^{10}$ cells, each cell containing $\widetilde{\Theta}(n^2)$ bits (note that the index of each cell can be specified by $10\log n$ bits). The input length then becomes $n^3\cdot 10\log n+\widetilde{\Theta}(n^{12})$. 
The resulting total function
\[
    \Gg_n:
    \{0,1\}^{n^3\cdot 10\log n+\widetilde{\Theta}(n^{12})}
    \to\{0,1\}
\]
evaluates to $1$ if and only if all $10\log n$ induced inputs to $\Ff_n$ lie in $\operatorname{Dom}(\Ff_n)$ and the cheat sheet cell indexed by their output string contains valid certificates for all the $\ANDOR$, proving both that the induced inputs satisfy the promise and that their function values agree with the index of the cell. The cheat sheet cells are part of the input and may therefore contain arbitrary strings.
If any certificate in the addressed cell is inconsistent with the corresponding $\ANDOR$ input, then it is invalid and $\Gg_n$ evaluates to $0$.

Theorem~5 and Lemma~6 of~\cite{Aaronson+STOC16} imply the following.

\begin{lemma}\label{lem:lb}
    $\RQ{\Gg_n}=\widetilde{\Omega} (n^2 \RQ{\Ff_n})$.
\end{lemma}

The best known separation between randomized and quantum query complexity was established independently by Sherstov, Storozhenko, and Wu~\cite{Sherstov+SICOMP23} and by Bansal and Sinha~\cite{BansalSinhaSTOC2021}, improving upon an earlier result of Tal~\cite{TalFOCS2020}.

\begin{lemma}[Corollary 1.2 in \cite{Sherstov+SICOMP23}; see also \cite{BansalSinhaSTOC2021}]
\label{lem:query-separation}
Let $t\ge 1$ be an integer.  There exists a partial function $f:D\to\{0,1\}$, where $D\subseteq\{0,1\}^n$, that can be computed using $t$ quantum queries with error at most $1/2-\Omega_t(1)$, while
\[
    \RQ{f}
    =
    \Omega\left(
        \frac{n^{1-\frac{1}{2t}}}
        {(\log n)^{2-\frac{1}{2t}}}
    \right).
\]
\end{lemma}

When $t=1$, the result above recovers, up to polylogarithmic factors, the separation of Aaronson and Ambainis~\cite{Aaronson+SICOMP18} based on the Forrelation problem: one quantum query achieves a constant advantage over random guessing, whereas the randomized query complexity is $\widetilde{\Omega}(\sqrt n)$. 

In the remainder of this paper, we let $\Ff_n$ be the partial function guaranteed by \cref{lem:query-separation}.

\paragraph{Query-to-communication lifting technique.}
To convert the (total) function~$\Gg_n$, defined above in the query complexity setting, into a (total) function defined in the communication complexity setting, we use the ``lifting'' technique from \cite{Chattopadhyay+ICAP19} based on the two-party computation of the inner product function $\IP_{m}\colon\{0,1\}^m\times \{0,1\}^m\to \{0,1\}$.

\begin{lemma}[\cite{Chattopadhyay+ICAP19}]\label{lem:lifting}
    For any function $f\colon\{0,1\}^s\to \{0,1\}$, there exists a constant $c$ such that  
    \[
        \RC{f\circ \IP_{c\log s}}=\Omega(\RQ{f}\log s).
    \]
\end{lemma}

\section{Proofs}

Throughout this section, we regard $t$ as a fixed constant. 

Let $N_n$ denote the input length of $\Gg_n$.  By construction,
\[
    N_n = n^3\cdot 10\log n+\widetilde{\Theta}(n^{12}) = \widetilde{\Theta}(n^{12}),
\]
and hence
\[
    \log N_n=\Theta(\log n).
\]
Let $c_0>0$ be the constant from \cref{lem:lifting}, and set
\[
    m := \left\lceil c_0\log N_n\right\rceil =
    \Theta(\log n).
\]
Applying $\IP_m$ coordinate-wise to the input of $\Gg_n$ gives the two-party communication problem
\[
    \Gg_n\circ\IP_m: (\{0,1\}^m)^{N_n} \times (\{0,1\}^m)^{N_n} \to\{0,1\}.
\]

Combining \cref{lem:lb,lem:query-separation,lem:lifting} gives the following lower bound.

\begin{theorem}\label{th:main-lb}
\[
    \RC{\Gg_n\circ\IP_m}
    =
    \widetilde{\Omega}
    \left(
        n^{3-\frac{1}{2t}}
    \right).
\]
\end{theorem}

\begin{proof}
By \cref{lem:lifting} and the choice of $m$,
\[
    \RC{\Gg_n\circ\IP_m}
    =
    \Omega\left(
        \RQ{\Gg_n}\log N_n
    \right).
\]
Using \cref{lem:lb,lem:query-separation} and $\log N_n=\Theta(\log n)$, we obtain
\[
\begin{aligned}
    \RC{\Gg_n\circ\IP_m}
    &=
    \widetilde{\Omega}
    \left(
        n^2\RQ{\Ff_n}
    \right)\\
    &=
    \widetilde{\Omega}
    \left(
        n^{3-\frac{1}{2t}}
    \right).
\end{aligned}
\]
\end{proof}

We next prove our main upper bound.

\begin{theorem}\label{th:main-ub}
\[
    \QCr{\Gg_n\circ\IP_m}{2t+2} = \widetilde{O}(n^2).
\]
\end{theorem}

\begin{proof}
Set
\[
    L:=10\log n.
\]
To compute $\Gg_n\circ\IP_m$, we perform the following three tasks:
\begin{enumerate}
    \item solve the $L$ instances of $\Ff_n\circ\ANDOR\circ\IP_m$ to obtain a tentative cheat sheet address;
    \item read the certificates stored at that address in the cheat sheet;
    \item verify the certificate data in the addressed cell against the instance part of the input.
\end{enumerate}

\paragraph{Implementation of Task 1.}

Fix one of the $L$ instances and suppress its instance index from the notation.  For each $j\in[n]$, write
\[
    x_j=(x_{j,r})_{r\in[n^2]}
    \quad\text{and}\quad
    y_j=(y_{j,r})_{r\in[n^2]},
\]
where
\[
    x_{j,r},y_{j,r}\in\{0,1\}^m.
\]
Define
\[
    z_j
    :=
    \ANDOR\left(
        \bigl(
            \IP_m(x_{j,r},y_{j,r})
        \bigr)_{r\in[n^2]}
    \right).
\]
Then
\[
    z=(z_1,\ldots,z_n)\in\{0,1\}^n
\]
is the input induced for $\Ff_n$.  Note that $z$ need not lie in $\operatorname{Dom}(\Ff_n)$.

A query to $z$ can be simulated coherently by a two-message Alice-to-Bob-to-Alice protocol.  Let $I$ be the query-index register and let $T$ be the target qubit.  Alice introduces an $mn^2$-qubit auxiliary register $M$ initialized to $\ket{0^{mn^2}}$ and applies the local unitary
\[
    \ket{j}_I\ket{u}_M
    \longmapsto
    \ket{j}_I\ket{u\oplus x_j}_M.
\]
In particular,
\[
    \ket{j}_I\ket{0^{mn^2}}_M
    \longmapsto
    \ket{j}_I\ket{x_j}_M.
\]
Alice sends the registers $I,M,T$ to Bob.  Regarding $u\in\{0,1\}^{mn^2}$ as
\[
    u=(u_r)_{r\in[n^2]},
    \qquad
    u_r\in\{0,1\}^m,
\]
Bob coherently applies
\[
\begin{aligned}
    \ket{j}_I\ket{u}_M\ket{b}_T
    \longmapsto
    \ket{j}_I\ket{u}_M
    \ket{
        b\oplus
        \ANDOR\left(
            \bigl(
                \IP_m(u_r,y_{j,r})
            \bigr)_{r\in[n^2]}
        \right)
    }_T.
\end{aligned}
\]
Bob sends all three registers back to Alice, who reverses her loading operation.  This returns $M$ to $\ket{0^{mn^2}}$ and exactly implements
\[
    O_z\ket{j}_I\ket{b}_T
    =
    \ket{j}_I\ket{b\oplus z_j}_T.
\]
Thus, one query to $z$ can be simulated using two messages and
\[
    2\bigl(
        mn^2+\lceil\log n\rceil+1
    \bigr)
\]
qubits of communication.

By \cref{lem:query-separation}, $\Ff_n$ admits a $t$-query quantum algorithm $\mathcal A$ that, on promised inputs, succeeds with probability at least
\[
    \frac12+\gamma
\]
for some constant $\gamma>0$ depending only on the fixed constant $t$.

We apply $\mathcal A$ to each of the $L$ instances.  For each instance, we run
\[
    R:=\Theta(\log n)
\]
independent copies of $\mathcal A$ and take the majority of their outputs.  By choosing the constant implicit in $R$ sufficiently large, a Chernoff bound reduces the error probability for each promised instance to $n^{-\Omega(1)}$.  A union bound over the $L$ instances then shows that all $L$ outputs are simultaneously correct with probability
\[
    1-n^{-\Omega(1)},
\]
provided that all the induced inputs satisfy the promise.

All $LR$ copies are run in parallel.  At each query step $s\in[t]$, the $s$-th oracle queries of all copies are simulated simultaneously within one Alice-to-Bob-to-Alice exchange. Therefore, the amplification does not increase the number of messages.

Let
\[
    a=(a_1,\ldots,a_L)\in\{0,1\}^L
\]
denote the tentative address obtained by Alice.  Task~1 uses $t$ Alice-to-Bob-to-Alice exchanges, hence $2t$ messages, and has total communication
\[
\begin{aligned}
    2tLR
    \bigl(
        mn^2+\lceil\log n\rceil+1
    \bigr)
    &=
    O\left(
        \log^2 n\,
        (mn^2+\log n)
    \right)\\
    &=
    \widetilde{O}(n^2),
\end{aligned}
\]
where we use $m=\Theta(\log n)$ and the fact that $t$ is constant.

\paragraph{Implementation of Task 2.}

Let
\[
    C=\widetilde{\Theta}(n^2)
\]
denote the number of bits in each cheat sheet cell.  For every
address $\alpha\in\{0,1\}^L$, let
\[
    c_\alpha
    =
    (c_{\alpha,\ell})_{\ell\in[C]}
    \in\{0,1\}^C
\]
denote the contents of the cheat sheet cell indexed by $\alpha$ before composition with $\IP_m$.

In the communication problem, for every $\alpha\in\{0,1\}^L$ and $\ell\in[C]$, Alice and Bob hold blocks
\[
    x_{\alpha,\ell},y_{\alpha,\ell}\in\{0,1\}^m,
\]
respectively, such that
\[
    c_{\alpha,\ell}
    =
    \IP_m(x_{\alpha,\ell},y_{\alpha,\ell}).
\]

Since Alice knows the tentative address $a$, she sends Bob $a$
together with all of her gadget inputs for the corresponding
cell,
\[
    X_a
    :=
    (x_{a,\ell})_{\ell\in[C]}.
\]
Bob selects the corresponding blocks
\[
    Y_a
    :=
    (y_{a,\ell})_{\ell\in[C]}
\]
from his input and reconstructs the addressed cell
coordinate-wise as
\[
    c_a
    =
    \bigl(
        \IP_m(x_{a,\ell},y_{a,\ell})
    \bigr)_{\ell\in[C]}.
\]
Thus, Bob reconstructs the entire contents of the addressed
cheat sheet cell.  This step consists of one message from Alice
to Bob and uses
\[
    L+mC
    =
    \widetilde{O}(n^2)
\]
bits of classical communication.

\paragraph{Implementation of Task 3.}

Bob first checks whether the reconstructed data $c_a$ are
syntactically well formed.  If not, he sends Alice a rejection
flag, and Alice rejects.  Otherwise, let $S_a$ denote the set of
coordinates in the instance part of the input whose values are
referred to by the certificates contained in $c_a$.

For every $q\in S_a$, Alice and Bob hold blocks
\[
    x_q,y_q\in\{0,1\}^m,
\]
respectively, such that the corresponding bit of the input before
composition with $\IP_m$ is
\[
    u_q=\IP_m(x_q,y_q).
\]
Bob sends Alice the certificate data $c_a$ together with all of
his gadget inputs corresponding to the coordinates in $S_a$,
namely,
\[
    (y_q)_{q\in S_a}.
\]
Alice combines them with her corresponding gadget inputs
\[
    (x_q)_{q\in S_a}
\]
and reconstructs all the input bits referred to by the
certificates:
\[
    (u_q)_{q\in S_a}
    =
    \bigl(
        \IP_m(x_q,y_q)
    \bigr)_{q\in S_a}.
\]

Alice checks the $\ANDOR$ certificates in $c_a$ against the reconstructed input bits. She accepts if and only if these certificates are valid, namely, if they certify that all the induced inputs lie in $\operatorname{Dom}(F_n)$ and that their output string is $a$.

Each of the $L$ instances of $\Ff_n\circ\ANDOR$ contains $n$
occurrences of $\ANDOR$, so there are $Ln$ such occurrences in
total.  Since a certificate for each occurrence refers to at
most $n$ input bits,
\[
    |S_a|
    \le
    Ln^2
    =
    \widetilde{O}(n^2).
\]
Moreover,
\[
    |c_a|
    =
    C
    =
    \widetilde{O}(n^2).
\]
Thus, Task~3 consists of one message from Bob to Alice and uses
\[
    |c_a|+m|S_a|
    =
    \widetilde{O}(n^2)
\]
bits of classical communication.

\paragraph{Correctness and complexity.}
Whenever Task~3 accepts, the verified certificates show that all the induced inputs satisfy the promise of $\Ff_n$ and that their output string is $a$. Hence $a$ is the correct cheat sheet address and the addressed cell contains valid certificate data, so the input is a $1$-input of $\Gg_n\circ\IP_m$. Conversely, on a $1$-input, Task~1 outputs the correct address with probability $1-n^{-\Omega(1)}$. Conditioned on this event, Task~2 reconstructs the addressed cell exactly, and Task~3 accepts. Therefore, the protocol computes $\Gg_n\circ\IP_m$ with bounded error.

Task~1 uses $2t$ messages, Task~2 uses one message from Alice to Bob, and Task~3 uses one message from Bob to Alice. Hence the
protocol uses at most
\[
    2t+2
\]
messages and
\[
    \widetilde{O}(n^2)
\]
qubits of communication in total.  
The classical messages in Tasks~2 and~3 can be sent as computational-basis qubits.
\end{proof}

\cref{th:main-lb,th:main-ub} have an implication for an exponential round separation between the classical and quantum two-node CONGEST model.

\begin{corollary}\label{cor:congest}
For every fixed integer $t\ge 1$, let $H_n:=\Gg_n\circ\IP_m$ be the total Boolean function constructed above.  There exists a bandwidth parameter $B_n=\widetilde{\Theta}(n^2)$ such that $H_n$ can be computed in at most $2t+2$ rounds in the two-node quantum $\mathrm{CONGEST}(B_n)$ model.  On the other hand, every bounded-error randomized classical $\mathrm{CONGEST}(B_n)$ algorithm for $H_n$ requires
\[
    \widetilde{\Omega}
    \left(
        n^{1-\frac{1}{2t}}
    \right)
\]
rounds.
\end{corollary}

\begin{proof}
The quantum upper bound follows directly from \cref{th:main-ub}: choose $B_n=\widetilde{\Theta}(n^2)$ large enough that each of the $2t+2$ messages fits into one distributed round.

Conversely, a $T$-round classical $\mathrm{CONGEST}(B_n)$ algorithm on a two-node network induces a two-party protocol of communication cost at most $2TB_n$.
Therefore, by \cref{th:main-lb},
\[
    2TB_n
    \ge
    \widetilde{\Omega}
    \left(
        n^{3-\frac{1}{2t}}
    \right).
\]
Since $B_n=\widetilde{\Theta}(n^2)$, we obtain
\[
    T
    =
    \widetilde{\Omega}
    \left(
        n^{1-\frac{1}{2t}}
    \right).
\]
\end{proof}
\section{Concluding remarks}

We study the power of the constant-round quantum communication complexity, and showed it is indeed powerful against the classical randomized communication complexity, where we allow any rounds two-way communication. Concurrent work by Gavinsky \cite{Gavinsky2026} showed it is even possible to show an exponential separation between them with 2-round quantum communication.

One open question is if we can get a similar quantum advantage with quantum one-way communication for a total function. Note that very recently, Anselm et al.~\cite{anselm2026optimal} showed an exponential separation between quantum one-way and classical two-way communication complexity for a total relation problem, and Anselm et al.~\cite{anselm2026exponential} showed an exponential separation between quantum one-way communication complexity with and without entanglement assistance for a total function.
Another interesting open question is if it is possible to get similar quantum advantage with a more natural function rather than a function using the cheat sheet technique.

\section*{AI disclosure}
All the ideas were derived by the authors, and all the results were obtained by them. They used ChatGPT 5.6 to assist in the write-up of the manuscript, and take full responsibility for the manuscript.

\section*{Acknowledgments}
The authors thank Richard Cleve and Amin Shiraz Gilani for discussions.

AH is supported by JSPS KAKENHI grants Nos.~24H00071, 25K24674, 25K24465. FLG is supported by JSPS KAKENHI grants Nos.~24H00071, 25K24674, 25K24465, MEXT Q-LEAP grant No.~JPMXS0120319794, JST ASPIRE grant No.~JPMJAP2302 and JST CREST grant No.~JPMJCR24I4.

\bibliographystyle{alpha}
\bibliography{References}
\end{document}